\documentclass[runningheads]{llncs}

\usepackage{biblatex}
\usepackage{graphicx}
\usepackage{wrapfig,fullpage}
\usepackage{booktabs,complexity,tabularx}
\usepackage{amsmath,amsfonts,amssymb}
\usepackage{stmaryrd}
\usepackage[boxed]{algorithm}
\usepackage[noend]{algorithmic}

\usepackage{turnstile}
\usepackage{mathrsfs}
\usepackage{enumitem}
\usepackage{amsthm}
\usepackage[super]{nth}
\usepackage{listings}
\usepackage{mathtools}
\usepackage{csquotes}
\usepackage{multicol,multirow}
\usepackage{fancyvrb}
\usepackage{makecell}
\usepackage{mdframed}
\usepackage[toc,page]{appendix}
\usepackage{mathdots}
\usepackage{empheq}
\usepackage{bm}
\usepackage{pgfplots}
\usepackage{enumitem}
\usepackage{float}
\usepackage{nicefrac} 
\usepackage{array}
\usepackage{blkarray}

\usepackage{tikz}
\usetikzlibrary{automata, positioning, calc, shapes, arrows, fit}

\newenvironment{claimproof}{\par\noindent\underline{Proof:}}{\leavevmode\unskip\penalty9999 \hbox{}\nobreak\hfill\quad\hbox{$\blacksquare$}}

\newcommand{\cgkernel}{\textsc{max-$c\gamma$-kernel}}
\newcommand{\cI}{\mathcal{I}}
\newcommand{\smti}{\textsc{Max-SMTI}}

\newcommand{\cminstab}{\textit{c$\gamma$-min stable}}
\newcommand{\cmaxstab}{\textit{c$\gamma$-max stable}}

\newcommand{\cgblocks}{$c\gamma$-blocks}
\newcommand{\cgstab}{$c\gamma$-stable}
\newcommand{\usmti}{\textsc{Max-$c\gamma$-smti}}

\newcommand{\mmatching}{common independent set}

\newcommand{\cgblocking}{$c\gamma$-blocking}
\newcommand{\pbDef}[3]{
     \begin{center}
     \colorbox{lightgray}{\begin{minipage}{0.95\linewidth}%
    \textsc{#1}
     \begin{itemize}[leftmargin=1.1cm]      
       \item[\textbf{In:}]  #2
       \item[\textbf{Out:}]  #3
     \end{itemize}
    \end{minipage}
    }
    \end{center}
  
}

\title{A Simple 1.5-approximation Algorithm for a Wide Range of Maximum Size Stable Matching Problems}

 \author{  Gergely Cs\'{a}ji\inst{1,2}}
 \institute{
 ELTE E\"{o}tv\"{o}s Lor\'and University, Budapest, Hungary
 \and  HUN-REN Centre for Economic and Regional Studies, Budapest, Hungary
}

\begin{document}
	\maketitle
\begin{abstract}
We give a simple approximation algorithm for a common generalization of many previously studied extensions of the maximum size stable matching problem with ties. These generalizations include the existence of critical vertices in the graph, amongst whom we must match as much as possible, free edges, that cannot be blocking edges and $\Delta$-stabilities, which mean that for an edge to block, the improvement should be large enough on one or both sides. We also introduce other notions to generalize these even further, which allows our framework to capture many existing and future applications. We show that the edge duplicating technique allows us to treat these different types of generalizations simultaneously, while also making the algorithm, the proofs and the analysis much simpler and shorter than in previous approaches. In particular, we answer an open question by Askalidis et al. \cite{socialstable} about the existence of a $\frac{3}{2}$-approximation algorithm for the \smti\ problem with free edges. This demonstrates that this technique can grasp the underlying essence of these problems quite well and have the potential to be able to solve many future applications. 
\end{abstract}

\section{Introduction}
Preference based matching markets is an extensively studied topic in computer science, mathematics and economics literature. The intensive study of the area started back in 1962 after the seminal paper of Gale and Shapley \cite{gale1962college}. They defined the model for the stable marriage problem and showed that a stable matching always exists and can be found in linear time. Since then, countless applications and related models have been studied, see Manlove \cite{manlove2013algorithmics} for an overview. Applications include Resident allocation, University admissions, Kidney exchanges, job markets and much more.

\subsection{Related work}\label{sec:related}
In the stable marriage problem, the goal is to match two classes of agents in a way such that there are no blocking edges, i.e. pairs of agents who mutually consider each other better than their partner.
The \emph{stable marriage problem with ties and incomplete lists} (\textsc{smti}) was first studied by Iwama et al. ~\cite{IMMM99}, who showed the NP-hardness of \textsc{Max-SMTI}, which is the problem of finding a maximum size stable matching in an \textsc{SMTI} instance. Since then, various algorithms have been proposed to improve the approximation ratio, e.g. Iwama et al. \cite{IMY07, IMY08} and Király \cite{Kiraly11}, and the current best ratio is $\frac{3}{2}$ by a polynomial-time algorithm of McDermid \cite{Mcdermid09}, where the same ratio is attained by linear-time algorithms of Pauluch \cite{paluch2011faster, Paluch14} and Király \cite{kiraly2012linear, Kiraly13}.
This $\frac{3}{2}$-approximation has been extended to critical relaxed stable matchings by a very recent paper of Nasre et al. \cite{critical-ties-approx}, which is a generalization of \smti, where there is a set of critical agents, amongst whom we must match as much as possible and stability is generalized in an appropriate way. Critical relaxed stability has also been studied in Krishnaa et al.  \cite{critical-strict-inapprox}, where they showed that a critical relaxed stable matching always exists, but finding a maximum size such matching is NP hard to approximate within $\frac{21}{19}-\varepsilon$, even with only strict preferences.

The approximation algorithm for \smti\ has also been extended to some cardinal and matroidal generalizations as well by Csáji et al. \cite{csaji2023approximation}. This include for example common quotas and $\Delta$-min-stability, where there is a number $\Delta$, and blocking is relaxed in a way such that only those edges can block, where both participating agents improve by at least $\Delta$. The matroidal generalization of stable matching was introduced in Fleiner \cite{fleiner2001matroid}.

As for the inapproximability of \textsc{Max-SMTI}, Halldórssson et al.  \cite{halldorsson2002inapproximability} showed that it  is NP-hard to approximate it within some constant factor. Later, inapproximability results have been improved, especially assuming stronger complexity theoretic conjectures. Yanagisawa \cite{yanagisawa2007approximation} and Halldórsson et al. \cite{yanigasawa2003improved} showed that assuming the Unique Games Conjecture, there is no $\frac{4}{3}-\varepsilon$-approximation for any $\varepsilon >0$. By a recent work by Dudycz et al. \cite{smti1.5inapprox} it follows that assuming the Small Set Expansion Hypothesis or the strong-UGC, there cannot even be a $\frac{3}{2}-\varepsilon$-approximation algorithm for \textsc{Max-SMTI}.

Stable matchings with free edges - that is, edges that can be included in a matching, but cannot block - have also been studied in many papers. Askalidis et al. \cite{socialstable} proved that with strict preferences and free edges, finding a maximum size stable matching can be $\frac{3}{2}$-approximated in polynomial time. However, on the negative side they showed that the problem does not admit a polynomial time $\frac{21}{19}-\varepsilon$ approximation algorithm unless $P= NP$ and assuming the Unique Games Conjecture, it does not even admit a $\frac{3}{2}-\varepsilon$ approximation. These negative results also hold with respect to strict preferences. They posed as an open question whether this $\frac{3}{2}$-approximation can be extended to weak preferences. In this paper we answer this question positively. Cechlárová and Fleiner \cite{cechlarova2009stable} investigated the stable roommates problem with free edges and proved NP-hardness even in very restricted settings. 

\subsection{Our Contributions}
In this paper we extend the $\frac{3}{2}$-approximation algorithm for \smti\ for a common generalization of many previously introduced and studied concepts. In particular, we answer an open question posed by Askalidis et al. \cite{socialstable} about the existence of a $\frac{3}{2}$-approximation algorithm for the \smti\ problem with free edges. We utilize a recently introduced edge duplicating technique by Yokoi \cite{yokoi2021approximation} and Csáji et al. \cite{csaji2023approximation} and show that with the help of this technique, we are able to solve much harder and more general related problems in a lot simpler and more elegant way. We give a simple approximation algorithm for a framework that includes free edges, critical vertices, critical edges, capacities, matroid constraints, many cardinal stability notions and more. First, we describe our algorithm for the one-to-one bipartite matching case in Section \ref{sec:one-to-one}. Then, we extend it to the most general case with matroidal constraints in Section \ref{sec:matroid}.

Although both our algorithm and our analysis are quite simple and elegant, we emphasize that our method can solve many problems that previously took whole papers in highly rated journals and conferences to tackle and also answers a decade long open question.

\section{Preliminaries}

We investigate matching markets, where the set of agents with the possible set of contracts is given by a bipartite graph $G=(U,W;E)$, where vertices correspond to agents and edges correspond to possible, mutually acceptable contracts. For each agent $v\in U\cup W$, let $E(v)$ denote the edges that are incident to $v$. We assume that for each agent $v$, there is a \textit{preference function} $p_v:E(v) \to \mathbb{R}_{\ge 0}$, which defines a weak ranking over the incident edges of $v$. We emphasize that we assume that the agents (vertices) rank their incident edges instead of the adjacent agents on the other side, because we allow parallel edges in our model, representing multiple types of contracts between two given agents. We also assume that $p_v(\emptyset )\le 0$, which denotes that an agent always weakly prefers to be matched to any acceptable partners rather than being unmatched. 

We say that an edge set $M\subseteq E$ is a \textit{matching}, if $|M\cap E(v)|\le 1$ for each $v\in U\cup W$. For a vertex $v$, let $M(v)$ denote the edge incident to $v$ in $M$, if there is any, otherwise $M(v)$ is $\emptyset$. 
We say that an edge $e=(u,w)\notin M$ \textit{blocks} a matching $M$, if $p_w(e)>p_w(M(w))$ and $p_u(e)>p_u(M(u))$ holds, that is, both agents strictly prefer each other to their partners in $M$. 
A matching $M$ is called \textit{weakly stable} or just \textit{stable}, if there is no blocking edge to $M$. The problem of finding a maximum size weakly stable matching in a bipartite graph with weak preferences, called \smti\ is a well studied problem that is NP-hard even to approximate within some constant, but admits a linear time $\frac{3}{2}$-approximation. 

We continue by introducing the generalizations that we consider. We start with $\Delta$-min and $\Delta$-max stability, which were introduced in Csáji et al. \cite{csaji2023approximation} (a similar notion to $\Delta$-min stability has also been defined in Chen et al. \cite{localstable}). The motivation behind these notions is that in many applications, blocking pairs may only arise if one or both member of the pair improves by a significant enough amount.

Let $\Delta >0$ be a constant.
A matching $M$ is $\Delta $-\textit{min-stable}, if there is no blocking edge $e=(u,w)$ with respect to $M$, such that $\min \{ p_u(e)-p_u(M(u)),p_w(e)-p_w(M(w))\} \ge \Delta$, that is, both agents improve by at least $\Delta$. The motivation behind this is clear: if a blocking does not result in a significant improvement, it may not even worth the effort. 

A matching $M$ is $\Delta $-\textit{max-stable}, if there is no blocking edge $e=(u,w)$ with respect to $M$, such that $\max \{ p_u(e)-p_u(M(u)),p_w(e)-p_w(M(w))\} \ge \Delta$, that is, one of the improvements is at least $\Delta$. Note here, that the edge $e$ still must be a blocking edge, so both $p_u(e)-p_u(M(u))>0$  and $p_w(e)-p_w(M(w))>0$ are assumed. The motivation behind this notion is that in many cases, it may be enough that one agent can improve by a sufficiently large amount to initiate a blocking contract, and the other agent only has to accept it, so for him even a smaller improvement is enough.
By choosing $\Delta$ to be small enough, it is easy to see that both of these stability notions strictly generalize weak stability.

Now we define critical matchings and relaxed stability, as they were defined in a very recent work of Krishnaa et al. \cite{critical-strict-inapprox}. For this, let us suppose that there is a set $C\subseteq U\cup W$ of \textit{critical} vertices. Criticality means that we want to match as many of them as possible in any matching. We say that a matching $M$ is \textit{critical}, if there is no other matching $M'$, such that $M'$ covers strictly more vertices from $C$, than $M$. Of course, there may not be a matching that is both critical and weakly stable. Hence, we need a weaker notion of stability, called \textit{relaxed stability}. A matching $M$ is \textit{relaxed stable}, if there is no blocking pair $e=(u,w)$ to $M$, such that $M(u)$ and $M(w)$ are not critical vertices. The motivation behind this definition is that since we want to maximize the number of matched critical agents, we do not allow blocking pairs, where one of the participating agents would leave a critical agent unmatched. We will slightly strengthen this notion of stability for the following reason: this does not allow a pair consisting of an unmatched critical agent and another agent whose partner is critical to block. However, after letting the agents of such a blocking pair to be matched together, the number of critical agents covered remains to same, (so maximal), hence such a pair indeed should be able to block. We will show that this strengthening of stability is still enough to guarantee existence.

Now we generalize the concept of criticality and relaxed stability in a natural way.
Suppose that there is a set $E_c$ of critical edges too. This represents that it may be also important that we match the critical agents with specific types of contracts, or to certain other agents that are good enough for them, or more compatible with them in some way. For a matching $M$, let $M_c=M\cap E_c$ denote the critical edges of $M$. We say that a matching $M$ is \textit{critical}, if there is no matching $N$, such that $N_c$ covers strictly more critical vertices than $M_c$. A matching $M$ is \textit{critical relaxed stable}, if it is critical and there is no blocking edge $e=(u,w)$, such that $M\setminus \{ (u,M(u)),(M(w),w)\} \cup \{ e\}$ is still critical. The special case of only critical vertices can be obtained by choosing $E=E_c$. 

Finally we introduce the notion of free agents and free edges. In Askalidis et al. \cite{socialstable}, free edges were used to model stable matching instances, where there are friendship and business relations which restrict certain edges from blocking. Let us suppose there is a set $F\subseteq E$ of free edges and a set $F_A\subseteq U\cup W$ of free agents. Free agents cannot participate in any blocking edge, while free edges cannot be blocking edges. Hence, a matching $M$ is \textit{stable} in this instance, if for each blocking edge $(u,w)$ it holds that $|\{ u,w\} \cap F_A|\ge 1$ or $(u,w)\in F$. Clearly, we can model free agents with free edges only: just make each edge adjacent to a free agent a free edge. 


Next we define a nice common generalization of $\Delta$-min/max stablity and the notion of free edges. Here we suppose that for each edge $e=(u,w)$, there are two numbers $\gamma_e^u>0$ and $\gamma_e^w>0$ given. We say that a matching is \textit{$\gamma$-min stable}, if there is no blocking edge $e=(u,w)$ such that $p_u(e)-p_u(M(u))\ge \gamma_e^u$ and $p_w(e)-p_w(M(w))\ge \gamma_e^w$ holds. Similarly, a matching is \textit{$\gamma$-max stable}, if there is no blocking edge $e=(u,w)$ such that $p_u(e)-p_u(M(u))\ge \gamma_e^u$ or $p_w(e)-p_w(M(w))\ge \gamma_e^w$ holds. The case of $\Delta$-stabilites with a set $F\subseteq E$ of free edges corresponds to the very special case, when $\gamma_e^v=\infty$, if $e\in F$ and $\gamma_e^v=\Delta$ otherwise. Hence, in our model, we allow very different types of conditions for each edge to block, independently from each other, which can incorporate many other special properties of certain applications. For example, if one thinks about job markets, and assumes that the underlying preferences are in correspondence with the salaries of the positions, there may be many other aspects of a workplace that make it desirable. Hence, for each agent and each different company and position, the increase in the salary that would make a job offer good enough for the applicant to switch, might be different.  Similary, depending on a company's existing employees and a new agent, it may depend on the specific skills of a new agent that are relevant for the company, how large of a salary within a contract would be worth it for the company to switch from an existing employee to the new one in a blocking pair.

Putting it all together, we define a general model, which incorporates all of the previously discussed ones. Suppose we are given a set $C\subseteq U\cup W$ of critical vertices, a set $E_c\subseteq E$ of critical edges and $\gamma_e^v$ values for each pair $(e,v)\in E\times (U\cup W)$ such that $v\in e$. We say that a matching $M$ is \cminstab, if $M$ is critical and there is no blocking edge $e=(u,w)$, such that $M\setminus \{ (u,M(u)),(M(w),w)\} \cup \{ e\}$ is still critical, and $\min \{ p_u(e)-p_u(M(u))-\gamma_e^u,p_w(e)-p_w(M(w))-\gamma_e^w\} \ge 0$ both hold. Otherwise, we call such a blocking edge a \textit{$c\gamma$-min blocking edge}. Clearly, this generalizes all three concepts: $\Delta$-min stability, free edges and critical vertices. Similarly we can define \cmaxstab\  matchings by replacing $\min \{ p_u(e)-p_u(M(u))-\gamma_e^u,p_w(e)-p_w(M(w))-\gamma_e^w\} \ge 0$ with $\max \{ p_u(e)-p_u(M(u))-\gamma_e^u,p_w(e)-p_w(M(w))-\gamma_e^w\} \ge 0$.

In order to be able to solve both problems with the same algorithm, let us go even one step further and consider a common generalization of the above two problems that we describe now. Instead of one $\gamma_e^u$, we have values $0<\gamma_e^u < \delta_e^u$ for each vertex-edge pair. We say that an edge $e=(u,w)$ \textit{\cgblocks} a matching $M$, if $M\setminus \{ (u,M(u)),(M(w),w)\} \cup \{ e\}$ is still critical, and either $\min \{ p_u(e)-p_u(M(u))- \gamma_e^u, p_w(e)-p_w(M(w))- \delta_e^w\} \ge 0$ or $\min \{ p_u(e)-p_u(M(u))- \delta_e^u,p_w(e)-p_w(M(w))-\gamma_e^w\} \ge 0$ holds. We say that a matching $M$ is \textit{\cgstab}, if $M$ is critical and there is no edge that \cgblocks\ $M$.

If the $\gamma_e^v$ values are sufficiently small, then this corresonds to $c\gamma$-max stability, and when the $\gamma_e^v$ and $\delta_e^v$ values are sufficently close, then it corresponds to $c\gamma$-min stability.

Now we define the optimization problem we investigate, called \textsc{Maximum  $c\gamma$-stable matching with ties, incomplete preferences and critical edges} abbriviated as \usmti\ for short. 
\pbDef{\usmti}{
A bipartite graph $G=(U,W;E)$, $p_v()$ preference valuations for each $v\in U\cup W$, a set $C\subseteq U\cup W$ of critical vertices, a set $E_c\subseteq E$ of critical edges, numbers $0<\gamma_e^v<\delta_e^v$ for each pair $(e,v)\in E\times (U\cup W)$ such that $v\in e$.
}{
A maximum size \cgstab\ matching $M$}

The main result of the paper is a simple $\frac{3}{2}$-approximation algorithm for \usmti.

\section{The Algorithm}\label{sec:one-to-one}
We start by describing a high level view of the algorithm.  Let $I$ be an instance of \usmti.

1. Create an instance $I'$ of the stable marriage problem with strict preferences by making parallel copies of each edge and create strict preferences over the created edges.  

2. Run the Gale-Shapley algorithm to obtain a stable matching $M'$ in the new instance $I'$

3. Take the projection $M$ of $M'$ to $I$ by taking an edge $e$ inside $M$, whenever one of the parallel copies of $e$ was inside $M'$.

The intuition behind this algorithm is the following. Our goal is a stable matching that is as large as possible. Hence, to help agents not to remain alone, we create copies for each edge, that allow rejected and unmatched agents multiple new chances to propose again to the same partner with a new, stronger contract that is more favorable to the other agent and can be good enough for her to reject another partner who has been considered better than the proposing agent before.

\subsection{The case of $E_c=E$ }

Let us start with the special case when each edge is critical, so only the number of covered critical vertices matters. Also, this means that $M_c=M$ for any matching $M$. The more general case will be solved by the same simple ideas, however, the many required types of copies make it more difficult to follow.
We describe the algorithm in more detail as follows. Let $I$ be an instance of \usmti. Let $|C\cap U|=s$ and $|C\cap W|=t$. For each edge $e=(u,w)\in E$ we create parallel edges as follows
\begin{itemize}
    \item [--] We create copies $a(e)$, $b_0(e),b_1(e)$ and $c(e)$,
    \item [--] If $w\in C\cap W$ is critical we create copies $x_1(e),\dots,x_t(e)$, 
    \item [--] If $u\in C\cap U$ is critical we create copies $z_1(e),\dots, z_s(e)$.
\end{itemize}
Then, we create strict preferences as follows. For a vertex $u\in U$, we rank the copies according to the rule 

$x_1 \succ \dots \succ x_t \succ a\succeq^{\gamma} b_0\succeq^{\delta -\gamma} b_1\succ c\succ z_s \succ \dots \succ z_1$. \\
For a vertex $w\in W$, we rank the copies according to the rule

$z_1\succ \dots \succ z_s\succ c\succeq^{\gamma}b_1\succeq^{\delta -\gamma} b_0\succ a\succ x_t\succ \dots \succ x_1$. \\
For any non-$b_i$ copy, we rank the edges of the same copy according to the preference functions $p_v()$ by breaking the ties arbitrarily. Here, $\alpha \succ \beta$ denotes that for any two edges $e,f$, the copy $\alpha (e)$ is ranked higher than the copy $\beta (f)$. 

The notation $a \succeq^{\gamma} b_0 \succeq^{\delta-\gamma}b_1$ means that for any edge $e$, $a(e)\succ_vb_0(e)\succ_vb_1(e)$, but $b_0 (f)\succ_v  a(e)$, if and only if $p_v(f)\ge p_v(e)+\gamma_f^v$  
and $b_1(f)\succ_va(e)$ if and only if $p_v(f)\ge p_v(e)+\delta_f^v$. This can be obtained in the following way. For each vertex $u\in U$ rank the $a$ copies according to $p_v()$ by breaking the ties arbitrarily. Then, we can insert the $b_0$ copy of each edge $f$, such that $b_0(f)\succ_ua(e)$, if and only if $p_u(f)\ge p_u(e)+\gamma_f^u$. Similarly we can insert the $b_1$ copies such that $b_1(f)\succ_ua(e)$ if and only if $p_u(f)\ge p_u(e)+\delta_f^u$. 
For the $W$ side, we do it similarly, according to their ranking.
Note that this definition allows the $b_i$ copies to be ranked in a different order than $p_v()$, which is surprising, but necessary. If there are still some ties remaining, we break them arbitrarily.

We remark another way to create the ranking over the $a,b_0,b_1,c$ copies. For a vertex $u\in U$ and $e=(u,w)$ let $p_u(a(e)) = p_u(e)$, $p_u(b_0(e)) = p_u(e)-\gamma_e^u$ and $p_u(b_1(e)) = p_u(e)-\delta_e^u$. Then, for each vertex, these new (possibly negative) preference number define a weak order, in which we break ties in a way such that with the same value the $b_1$ copies are best, the $b_0$ copies are second and the $a$ copies are last. We can do it similary for the $W$ side, with the tiebreaking $b_0$ first, $b_1$ second, $c$ best. 

We illustrate the creation of the extended preferences on a concrete example with the second method. Let $s=t=1$. Let $u\in U$ be an agent with three adjacent edges $e,f,g$. Let $p_u(e) =1$, $p_u(f)=3$ and $p_u(g)=4$. Also, let $\gamma_e^u=1,\gamma_f^u=2, \gamma_g^u=2$ and $\delta_e^u=2,\delta_f^u=3,\delta_g^u=6$. Then,
$p_u(a(g)) = 4, p_u(a(f)) = 3, p_u(b_0(g)) = 2, p_u(b_0(f))=p_u(a(e))=1, p_u(b_1(f))=p_u(b_0(e))=0, p_u(b_1(e))=-1, p_u(b_1(g))=-2$.
Then, the ranking between the copies is the following:

$x_1(g)\succ_u x_1(f)\succ_u x_1(e)\succ_u a(g) \succ_u a(f) \succ_u b_0(g)  \succ_u b_0(f) \succ_u a(e) \succ_u b_1(f) \succ_u b_0(e) \succ_u b_1(e) \succ_u b_1(g) \succ_u c(g) \succ_u c(f) \succ_u c(e) \succ_u z_1(g) \succ_u z_1(f) \succ_u z_1(e)$.

Notice that here $b_1(f)\succ_u b_1(e)\succ_u b_1(g)$, so for the $b_1$ copy, the edges are not ranked in their original order.

\begin{theorem}
\usmti\ can be $\frac{3}{2}$-approximated in $\mathcal{O}((s+t)|E|)$ time, if $E=E_c$.
\end{theorem}
\begin{proof}
We prove the theorem in three simple claims. Let $M$ denote the output of the algorithm and $M'$ be its preimage in the extended instance $I'$.

\begin{claim}
$M$ is critical.
\end{claim}
\begin{claimproof}
Suppose for contradiction that there is a matching $N$, such that $N$ covers strictly more critical vertices than $M$. Then, it must hold that either $N$ covers more critical vertices from $U$ or more critical vertices form $W$ than $M$. 

Suppose that the first case holds. Then, there is a component $P$ (a path) in $N\cup M$, such that there is an endpoint $u_1\in U\cap C$ that is critical, but is only matched in $N$, and also more critical vertices are covered in $N\cap P$, than in $M\cap P$. Let $e_1=(u_1,w_1)=N(u_1)$. As $u_1$ is critical, we know that the copies $z_1(u_1,w_1),\dots, z_s(u_1,w_1)$ exist. As $u_1$ is not matched in $M'$, the fact that $z_1(u_1,w_1)$ does not block $M'$ implies that $w_1$ is matched in $M'$ and with a $z_1$-type edge. Let $z_1(u_2,w_1)=M'(w_1)$. This immediately implies that $u_2\in U\cap C$, as $z_1(u_2,w_1)$ exists. By our assumption on the component (more critical vertices are covered in $N\cap P$, than in $M\cap P$), $u_2$ is matched in $N$ to a vertex $w_2$.

As any $z_2$ copy is better for $u_2$ in $I'$, but $z_2(u_2,w_2)$ exists and does not block, $z_j(u_3,w_2)\in M'$ for some $u_3\in U$ and $j\le 2$. Again, we obtain that $u_3$ is critical, so by our assumption on the component, $u_3$ is matched in $N$ to a vertex $w_3$. By iterating this argument, we get that there must be vertices $u_1,u_2,\dots , u_{s+1}\in U$ that are critical, which contradicts $|U\cap C|=s$.

The second case is analogous. 
\end{claimproof}

\begin{claim}
The output matching $M$ by the algorithm is \cgstab.

\end{claim}
\begin{claimproof}
Suppose for contradiction that there is a \cgblocking\ edge $e=(u,w)$ to $M$. 

If $u$ is unmatched, then $g=(u',w)=M(w)$ exists. As $e$ \cgblocks\ and $M$ is critical, we get that $u\in C$ if and only if $u'\in C$. Hence, the best copy of $e$ and $g$ for $w$ has the same type, either $z_1(e)$ or $c(e)$. In both cases, we get that this copy of $e$ blocks $M'$ as these copies are both ranked according to $p_w()$, contradiction. If $w$ is unmatched, $x_1(e)$ or $a(e)$ blocks $M'$ for similar reasons, contradiction.

Hence, $M(u)=f=(u,w')$ and $M(w)=g=(u',w)$ is not $\emptyset$. Because $e$ \cgblocks, $p_u(e)\ge p_u(f)+\gamma_e^u$ and $p_w(e)\ge p_w(g)+\delta_e^w$ or $p_u(e)\ge p_u(f)+\delta_e^u$ and $p_w(e)\ge p_w(g)+\gamma_e^w$. 
Also, as $e$ \cgblocks, the vertices $u',w'$ cannot be critical, so $x_1(f),\dots,x_t(f),z_1(g),\dots,z_s(g)$ do not exist. Hence, in the first case $b_0(e)$ and in the second case $b_1(e)$ blocks $M'$, contradiction.

\end{claimproof}

\begin{claim}
    For any \cminstab\ matching $N$ it holds that $|N|\le \frac{3}{2}|M|$.
\end{claim}
\begin{claimproof}
    Suppose for contradiction that there is a \cgstab\ matching $N$ such that $|N|>\frac{3}{2}|M|$. Then, there must be a path component in $N\cup M$ that is either a single $N$-edge, or consists of two $N$-edges and one $M$-edge. The first case is clearly impossible, as $M$ is necessarily maximal. 

    Suppose the second case holds. Let $e=(u_1,w_2)$ be the edge of $M$ and $f=(u_1,w_1)$, $g=(u_2,w_2)$ be the edges of $N$. 

    First observe that as $M$ and $N$ are critical, $e$ and $f\cup g$ cover the same number of critical vertices, so $w_1,u_2$ are not critical and $x_1(f),\dots,x_t(f),z_1(g),\dots,z_s(g)$ does not exists.

    As $a(f)$ does not block, $b_0(e),b_1(e)$, $a(e)$ or $x_i(e)\in M'$ for some $i$. As $c(g)$ does not block, $b_0(e),b_1(e)$,$c(e)$ or $z_j(e)\in M'$ for some $j$. Hence, $b_0(e)$ or $b_1(e)\in M'$. If $b_0(e)\in M'$, then $p_{u_1}(e)\ge p_{u_1}(f)+\gamma_e^{u_1}$ and $p_{w_2}(e)\ge p_{w_2}(g)+\delta_e^{w_2}$ as $M'$ is stable, so $e$ \cgblocks\ $N$ (combining with the fact that $e$ and $f\cup g$ cover the same number of critical vertices), contradiction. The other case implies $p_{u_1}(e)\ge p_{u_1}(f)+\delta_e^{u_1}$ and $p_{w_2}(e)\ge p_{w_2}(g)+\gamma_e^{w_2}$, so $e$ \cgblocks\ $N$, contradiction again.
    


\end{claimproof}

The statement follows from these three claims. The running time of the algorithm is linear in the number of edges of the extended instance, as the Gale-Shapley algorithm is linear, so it has running time $\mathcal{O}((s+t)|E|)$.
\end{proof}

\subsection{General case}

Next, we discuss the more general case with an arbirtary subset of critical edges.

We may assume without loss of generality that each critical edge $E_c$ is adjacent to at least one critical vertex, otherwise it cannot cover any critical vertices, so we get an equivalent instance by making the edge non-critical.

Let us start by describing the extended instance $I'$. For each edge $e=(u,w)\in E$ we create parallel edges as follows:
\begin{itemize}
    \item [--] We create copies $a(e)$, $b_0(e),b_1(e)$ and $c(e)$,
    \item [--] If $w\in C\cap W$ is critical and $e\in E_c$ is also critical, we create copies $x_1(e),\dots,x_{t+7}(e)$
    \item [--] If $u\in C\cap U$ is critical and $e\in E_c$ is also critical, we create copies $z_1(e),\dots, z_{s+7}(e)$
    \item [--] Finally, if $e\in E_c$ is critical and both $u,w\in C$ are critical, we create a copies $y_0(e),y_1(e)$
\end{itemize}

Then, we create strict preferences as follows. Let $\delta' = \delta -\gamma$. For a vertex $u\in U$, we rank the copies according to the rule

$x_1\succeq^{\gamma}x_2\succeq^{\delta'}x_3\succ x_4\succ \dots \succ x_{t+4}\succ z_{s+7}\succeq^{\gamma}(y_0\succeq z_{s+6})\succeq^{\delta'} (y_1\succeq z_{s+5})\succ z_{s+4}\succ \dots \succ z_4\succ a\succeq^{\gamma}(b_0\succeq z_3\succeq x_{t+5})\succeq^{\delta'}( b_1\succeq z_2\succeq x_{t+6})\succ z_1\succ x_{t+7}\succ c$. \\
For a vertex $w\in W$, we rank the copies according to the rule 

$z_1\succeq^{\gamma}z_2\succeq^{\delta'}z_3\succ z_4\succ \dots \succ z_{s+4}\succ x_{t+7}\succeq^{\gamma}(y_1\succeq x_{t+6})\succeq^{\delta'} (y_0\succeq x_{t+5})\succ x_{t+4}\succ \dots \succ x_4\succ c\succeq^{\gamma}(b_1\succeq x_3\succeq z_{s+5})\succeq^{\delta'} (b_0\succeq x_2\succeq z_{s+6})\succ x_1\succ z_{s+7}\succ a$. \\

Here, when we denote $(\alpha \succeq \beta)$ for some types of copies, then it always means that we insert $\alpha (e)$ and $\beta (e)$ together with $\beta (e)$ being strictly after $\alpha (e)$ (so no other edge comes between them) and we keep them together, so in the end for any other edge, either it is worse than both of them or better than both of them. In the model where we assign preference values to the copies to create a new ranking, this corresponds to that these copies always get the same value, so they are tied, and then we keep them together in the tiebreaking in that order.

The only important thing is that, if we have some copies $\alpha \succeq^{\gamma} ( \beta_1 \succeq \cdots \succeq \beta_k) \succeq^{\delta'} (\omega_1 \succeq  \cdots \succeq \omega_l )$, then for any edge $e$, $\alpha (e) \succ_v \beta_1(e) \succ_v \cdots \succ_v \beta_k(e) \succ_v \omega_1(e) \succ_v \cdots \succ_v \omega_l(e)$ and that for any $i\in [\{ 1,\dots,k\}$ and $j\in \{ 1,\dots,l\}$ it holds that $\beta_i (f) \succ_v \alpha(e)$, if and only if $p_v(f)\ge p_v(e)+\gamma_f^v$ and $\omega_j (f)\succ_v \alpha (e)$ if and only if $p_v(f)\succ p_v(e)+\delta_f^v$.

The edges from copies $x_1,x_4,x_5,\dots,x_{t+4},x_{t+7},z_1,z_4,z_5\dots,z_{s+4},$ $z_{s+7},a,c$ are ranked according to $p_v()$, the others are then inserted to the right place as before. In the preference value model, we do this by computing the corresponding values similarly as in the previous case (with each copy in a parenthesis getting the same value) and for tie breaking we use a suitable rule. For example for a vertex $u\in U$, amongst the copies $a,b_0,b_1,z_2,z_3,x_{t+5},x_{t+6}$, we can use $b_1\succ_u z_2\succ_u x_{t+6}\succ_u b_0\succ_u z_3\succ_u x_{t+5} \succ_u a$. The ties between the same type of copies can be broken arbitrarily.

\begin{theorem}
\usmti\ can be $\frac{3}{2}$-approximated in $\mathcal{O}((s+t)|E|)$ time.
\end{theorem}
\begin{proof}
We prove the theorem in three simple claims. Let $M$ denote the output of the algorithm and $M'$ be its preimage in the extended instance $I'$.

\begin{claim}
$M$ is critical.
\end{claim}
\begin{claimproof}
Suppose for contradiction that there is a matching $N$, such that $N_c$ covers strictly more critical vertices than $M_c$. Then, it must hold that either $N_c$ covers more critical vertices from $U$ or more critical vertices from $W$ than $M_c$. 

Suppose that the first case holds. Then, there is a component (a path) in $N_c\cup M_c$, such that there is an endpoint $u_1\in U\cap C$ that is critical, but is only covered in $N_c$ and also more critical vertices are covered in $N_c\cap P$, than in $M_c\cap P$. Let $e_1=(u_1,w_1)=N_c(u_1)$. As $u_1$ is critical, we know that the copies $z_1(u_1,w_1),\dots, z_{s+7}(u_1,w_1)$ exist. As $u_1$ is not covered with a critical edge in $M$ (as it is uncovered in $M_c$), it can only be matched with an $a,b_i$ or $c$ type copy. The fact that $z_4(u_1,w_1)$ does not block $M'$ implies that $w_1$ is matched in $M'$ and with a $z_j$-type edge with $j\le 4$. Let $(u_2,w_1)=M(w_1)$. This immediately implies that $u_2\in U\cap C$ and $(u_2,w_1)\in E_c$, as $z_j(u_2,w_1)$ exists. By our assumption on the component (more critical vertices are covered in $N_c\cap P$, than in $M_c\cap P$), $u_2$ is matched in $N_c$ to a vertex $w_2$.

As any $z_5$ copy is better for $u_2$ in $I'$ than any $z_j$ copy with $j\le 4$, but $z_5(u_2,w_2)$ exists and does not block, $z_j(u_3,w_2)\in M'$ for some $u_3\in U$ and $j\le 5$. Again, we obtain that $u_3$ is critical and $(u_3,w_2)\in E_c$, so by our assumption on the component, $u_3$ is matched in $N_c$ to a vertex $w_3$. By iterating this argument, we get that there must be vertices $u_1,u_2,\dots , u_{s+1}\in U$ that are critical, which contradicts $|U\cap C|=s$.

The second case is analogous. 
\end{claimproof}

\begin{claim}
The output matching $M$ by the algorithm is \cgstab.

\end{claim}
\begin{claimproof}
Suppose for contradiction that there is a \cgblocking\ edge $e=(u,w)$ to $M$. 

If $u$ is unmatched, then $g=M(w)=(u',w)$ exists. As $e$ \cgblocks\ and $M$ is critical we get that $e\in E_c$ if and only if $g\in E_c$ and if yes, then $u'\in C$ if and only if $u\in C$. In particular, we get that the best copy of $e$ and $g$ in $I'$ for $w$ is the same. Also, this copy can be only $c,x_{t+7}$ or $z_1$ and all these copies are ranked according to $p_w()$. Therefore, this copy of $e$ blocks $M'$, as $e$ \cgblocks\ $M$, contradiction. The case when $w$ is unmatched is similar.

Hence, $M(u)=f=(u,w')$ and $M(w)=g=(u',w)$ is not $\emptyset$. Because $e$ \cgblocks, $p_u(e)\ge p_u(f)+\gamma_e^u$ and $p_w(e)\ge p_w(g)+\delta_e^w$ or $p_u(e)\ge p_u(f)+\delta_e^u$ and $p_w(e)\ge p_w(g)+\gamma_e^w$. 

If $f,g\notin E_c$, then $e\notin E_c$, as $M$ is critical. Hence, only $a,b_0,b_1,c$ copies exists of $e,f,g$. As $e$ \cgblocks\ $M$, we get that $b_0(e)$ or $b_1(e)$ blocks $M'$, contradiction.

Suppose that $f,g\in E_c$. Then, $e\in E_c$, $u,w\in C$ and $u',w'\notin C$, because $e$ \cgblocks. Hence, $x_i(f),z_j(g)$ do not exists for any $i,j$, but $y_0(e),y_1(e)$ exist. As $e$ \cgblocks\, one of $y_0(e),y_1(e)$ blocks $M'$, contradiction.

Suppose that $f\in E_c,g\notin E_c$. In particular, only $a(g),b_0(g),b_1(g),c(g)$ exist of $g$. Then, $e\in E_c$ as $e$ \cgblocks\ and $w'\in C$ if and only if $w\in C$. If $w,w'\in C$, then $x_2(e),x_3(e)$ exist and one of them blocks $M'$, contradiction. If $w,w'\notin C$, then $u\in C$ (by our assumption for any $e\in E_c$, at least one endpoint is critical) and $x_i(e),x_i(f)$ do not exists for any $i$ and neither does $y_0(e),y_0(f),y_1(e),y_1(f)$. Hence, $z_{s+5}(e)$ or $z_{s+6}(e)$ blocks $M'$, contradiction.

The case $f\notin E_c,g\in E_c$ is analogous. 

\end{claimproof}

\begin{claim}
    For any \cminstab\ matching $N$ it holds that $|N|\le \frac{3}{2}|M|$.
\end{claim}
\begin{claimproof}
    Suppose for contradiction that there is a \cgstab\ matching $N$ such that $|N|>\frac{3}{2}|M|$. Then, there must be a path component in $N\cup M$ that is either a single $N$-edge, or consist of two $N$-edges and one $M$-edge. The first case is clearly impossible, as $M$ is necessarily maximal. 

    Suppose the second case holds. Let $e=(u_1,w_2)$ be the edge of $M$ and $f=(u_1,w_1)$, $g=(u_2,w_2)$ be the edges of $N$. 

    First observe that as $M$ and $N$ are critical, $e$ and $f\cup g$ cover the same number of critical vertices with critical edges. Hence, if $e$ $\gamma$-blocks $N$, then it \cgblocks\ $N$.

Suppose that $f,g\notin E_c$. Then, $e\notin E_c$ and only $a,b_0,b_1,c$ copies exists of all three edges. As $a(f)$ does not block $M'$, $b_0(e),b_1(e)$ or $a(e)\in M'$. As $c(g)$ does not block $M'$, $b_0(e),b_1(e)$ or $c(e)\in M'.$ Hence, $b_0(e)$ or $b_1(e)\in M'$. Furthermore, if $b_0(e)\in M'$, then $p_{u_1}(e)\ge p_{u_1}(f)+\gamma_e^{u_1}$ and $p_{w_2}(e)\ge p_{w_2}(g)+\delta_e^{w_2}$, so $e$ \cgblocks\ $N$ (combining with the fact that $e$ and $f\cup g$ cover the same number of critical vertices), contradiction. The other case implies $p_{u_1}(e)\ge p_{u_1}(f)+\delta_e^{u_1}$ and $p_{w_2}(e)\ge p_{w_2}(g)+\gamma_e^{w_2}$, so $e$ \cgblocks\ $N$, contradiction again.

Next suppose that $f,g\in E_c$. Then, as $M$ is critical, $e\in E_c$, $u_1,w_2\in C$ and $u_2,w_1\notin C$. Hence, $x_i(f),z_j(g)$ does not exists for any $i,j$ but $y_0(e),y_1(e)$ exist and also $z_j(f),x_i(g)$ exists for all $i,j$. As $z_{s+7}(f)$ does not block and $x_{t+7}(g)$ does not block $M'$ we get that $y_0(e)\in M'$ or $y_1(e)\in M'$ and in both cases, $e$ \cgblocks\ $N$, contradiction.

Suppose that $f\in E_c,g\notin E_c$. Then, as $M,N$ are critical, $e\in E_c$ and $w_1\in C$ if and only if $w_2\in C$. If $w_1,w_2\in C$, then $x_1(f)$ exists. As $x_1(f)$ does not block and $c(g)$ does not block $M'$, $x_2(e)\in M'$ or $x_3(e)\in M'$. In both cases, $e$ \cgblocks\ $N$, contradiction. If $w_1,w_2,\notin C$, then $u_1\in C$ (each $e\in E_c$ has one critical endpoint by our assumption) and $y_0(e),y_1(e),x_i(e)$ does not exists for $i\in \{ 1,\dots, t+7\}$, but $z_j(e),z_j(f)$ does for $j\in \{ 1,\dots, s+7\}$. As neither $z_{s+7}(f)$ nor $c(g)$ blocks $M'$, we get that $z_{s+6}(e)\in M'$ or $z_{s+5}(e)\in M'$. In both cases, we get that $e$ \cgblocks\ $N$, contradiction.

The case $f\notin E_c,g\in E_c$ is similar.

\end{claimproof}

The statement follows from these three claims. The running time of the algorithm is linear in the edges of the extended instance, as the Gale-Shapley algorithm is linear, so it has running time $\mathcal{O}((s+t)|E|)$.
\end{proof}

We state one remark about a straightforward extension of this algorithm.

\begin{remark}
If we can have $k$ $(\gamma_e^v)_1<(\gamma_e^v)_2<\dots <(\gamma_e^v)_k$ values instead of the two $\gamma_e^v,\delta_e^v$ values, and let $e=(u,w)$ block if $p_u(e)\ge p_u(M(u))+(\gamma_e^u)_i$ and $p_w(e)\ge p_w(M(w))+(\gamma_e^w)_{k+1-i}$ for some $i\in \{ 1,\dots,k\}$, then this framework can also incorporate a notion of $\gamma$-sum stability. In $\gamma$-sum stability, we have one $\gamma_e$ value for each edge, and $(u,w)$ $\gamma$-sum blocks, if it blocks and the sum of the two improvements are at least $\gamma_e$. With $k=\mathcal{O}(|U\cup W|^2)$ copies, we can have $(\gamma_e^v)_1<(\gamma_e^v)_2<\dots <(\gamma_e^v)_k$ values for each possible improvement $c$ and $\gamma_e-c$ such that $(\gamma_e^v)_i=\gamma_e-(\gamma_e^v)_{k+1-i}$, so $\gamma$-sum stability becomes a special case. 
Furthermore, the same algorithm straightforwardly extends to this case, we only need to make $k$ copies instead of $2$ for the needed cases and define the strict ranking in $I'$ according to the differences $(\gamma_e^v)_{i+1}-(\gamma_e^v)_i$.
\end{remark}

\section{General Case with Matroid Constraints}\label{sec:matroid}
Let us also extend the framework for two sided matroid constraints (which generalizes capacities and laminar common quotas too). For this, we first define matroids formally.

\subsection{Matroids}
A {\em matroid} is a pair $(E, \cI)$ of a finite ground set $E$ and a nonempty family $\cI\subseteq 2^E$ satisfying the following two axioms: (i) $A\subseteq B\in \cI$ implies $A\in \cI$, and (ii) for any $A, B\in \cI$ with $|A|<|B|$, there is an element $x\in B\setminus A$ with $A+x\in \cI$. A set in $\cI$ is called an {\em independent set}, and an inclusion-wise maximal one is called a {\em base}. By axiom (ii), all bases have the same size, which is called the {\em rank} of the matroid.
A {\em circuit} is an inclusion-wise minimal dependent set. The \emph{fundamental circuit} of an element $x\in S\setminus B$ for a base $B$, denoted by $C_B(x)$, is the unique circuit in $B+x$. By a slight abuse of notation, we will also use $C_I(x)$ for an independent set $I$ and an element $x\in E\setminus I$ to denote the unique circuit in $I+x$ if it exists. 


We define some simple operations of matroids that also produce matroids. 
Let $\mathcal{M}_1=(E_1,\cI_1), \mathcal{M}_2=(E_2,\cI_2)$ be two matroids. The {\em direct sum} $\mathcal{M}_1\oplus \mathcal{M}_2$ is the matroid, where the ground set is $E_1\cup E_2$ and a set $I\subseteq E_1\cup E_2$ is independent, if and only if $I\cap E_1\in \cI_1$ and $I\cap E_2\in \cI_2$.

The {\em truncation} of a matroid $\mathcal{M}=(E,\cI)$ to size $k$ is the matroid, where the ground set is $E$ and a set $I\subseteq E$ is independent if and only if $I\in \cI$ and $|I|\le k$. For a subset $J\subseteq E$, the {\em deletion of $J$} defines the matroid $\mathcal{M}\setminus J$, whose ground set is $E\setminus J$ and $I\subseteq E\setminus J$ is independent if and only if it is independent is $\mathcal{M}$. Finally, the {\em contraction of an independent set $J$} defines a matroid $\mathcal{M}/J$, whose ground set is similarly $E\setminus J$, and a set $I$ is independent, if and only if $I\cup J\in \cI$.

If we have a total order $\succ$ given on $E$, then the triple $M=(E,\cI,\succ)$ is called a \emph{(totally) ordered matroid}. A nice property of totally ordered matroids is that for any weight vector $w\in \mathbb{R}^E$ that satisfies $w_x>w_y \Leftrightarrow x \succ y$, the unique maximum weight base is the same. We call this base $A$ the \emph{optimal base} of $(E,\cI,\succ)$; it is characterized by the property that the worst element of $C_A(x)$ is $x$ for any $x \in S \setminus A$. 
We need the following theorem about ordered matroids from \cite{csaji2023approximation}. 
\begin{theorem}{(Csáji, Király, Yokoi)}
\label{thm:matching}
Let $\mathcal{M}=(E,\mathcal{I},\succ)$ be an ordered matroid of rank $r$. Let $A$ be the optimal base and $B$ be a base disjoint from $A$. Then, there is a perfect matching $a_ib_i$ $(i\in [r])$ between $A$ and $B$ such that $a_i\succ b_i$ and $B+a_i-b_i \in \mathcal{I}$ for every $i \in [r]$. 
\end{theorem}

We also need a well-known simple Lemma about matroids.
\begin{lemma}
\label{matroidfact}
    Let $I,J$ be two independent sets of a matroid $\mathcal{M}$ such that for each $e\in I$, $J+e\notin \mathcal{M}$. Then, for each subset $I'\subseteq I$ it holds that the union of the fundamental cycles of the elements in $I'$ contains at least $|I'|$ elements.
\end{lemma}
\begin{proof}
If for a subset $I'$, the union of the fundamental cycles of these elements $C_{I'}$ in $J$ would be smaller than $|I'|$, then by axiom (ii) of matroids, there would be an element $e\in I'$ such that $C_{I'}+e\in \cI$, which contradicts that the fundamental cycle of $e$ is in $C_{I'}$.
\end{proof}

An important class of matroids is {\em laminar matroids}. A matroid $\mathcal{M}=(E,\cI )$ is a {\em laminar matroid}, if there is a laminar family $\mathcal{S}$ of sets over $E$, together with integer numbers $q_S\in \mathbb{N}$ for each $S\in \mathcal{S}$, such that a subset $I\subseteq E$ is independent, if and only if $|I\cap S|\le q_S$ for each $S\in \mathcal{S}$. In particular, if a matroidal constraint comes from a simple vertex capacity $q(v)$, then it defines a laminar matroid with $\mathcal{S} = \{ E(v)\}$ and $q_S=q(v)$. 

\subsection{Extending the model to matroid constraints}

We continue by extending our model and stability notions for the case of matroidal constraints.
Here, for each agent $v\in U\cup W$, apart from the preference valuation $p_v()$, there is also a matroid $\mathcal{M}_v=(E_v,\cI_v)$ over $E(v)$. For each $v\in U\cup W$, let $r(v)$ be the rank of the matroid $\mathcal{M}_v$. By taking the direct sums of the matroids on both sides, and appending the preference lists of the vertices, we create two matroids $\mathcal{M}_U=(E,\cI_U)$ and $\mathcal{M}_W=(E,\cI_W)$.
Here, $M\subseteq E$ is called a \textit{\mmatching}, if $M\cap E(v)$ is independent in the matroid $\mathcal{M}_v$ for each $v\in U\cup W$ or equivalently, $M$ is independent in both $\mathcal{M}_U$ and $\mathcal{M}_W$.

Let $C\subseteq U\cup W$ denote the set of critical agents and $E_c\subseteq E$ the set of critical edges and suppose we are given numbers $0<\gamma_e^v<\delta_e^v$ for each edge $e$ and vertex $v\in e$ as before. Without loss of generality, we suppose that for each critical edge $e$, at least one of its endpoints is critical.
Again, let $M_c=M\cap E_c$.
We say that a \mmatching\ $M$ is \textit{critical}, if $M_c$ fills at least as many places of critical vertices as $N_c$ for any other \mmatching\ $N$. That is, $\sum_{v\in C}|M(v)\cap E_c|$ is maximum among all common independent sets.

We say that an edge $e=(u,w)$ \textit{\cgblocks} a \mmatching\ $M$, if there are $f\in M(u)\cup \{ \emptyset \},g\in M(w)\cup \{ \emptyset \}$ such that $M'=M\setminus \{ f,g\} \cup \{ e\}$ is also critical, $M-f+e\in \cI_U$, $M-g+e\in \cI_w$ and either $p_u(e)-p_u(f)\ge \gamma_e^u$ and $p_w(e)-p_w(g)\ge \delta_e^w$ or $p_u(e)-p_u(f)\ge \delta_e^u$ and $p_w(e)-p_w(g)\ge \gamma_e^w$ holds. $M$ is \textit{\cgstab}, if no edge \cgblocks\ $M$.

Hence, we obtain the following optimization problem:

\pbDef{\cgkernel}{
A bipartite graph $G=(U,W;E)$, $p_v()$ preference valuations for each $v\in U\cup W$, a set $C\subseteq U\cup W$ of critical vertices, a set $E_c\subseteq E$ of critical edges, matroids $\mathcal{M}_v$ for each $v\in U\cup W$ and numbers $0<\gamma_e^v<\delta_e^v$ for each pair $(e,v)\in E\times (U\cup W)$ such that $v\in e$.
}{
A maximum size \cgstab\ common independent set $M$.}

\subsection{Extending the algorithm}

Let $s=\sum_{u\in U\cap C}r(u)$ and $t=\sum_{w\in W\cap C}r(w)$ be the sum of ranks of the matroids of the critical agents in $U$ and $W$ respectively. In particular, $s,t$ are upper bounds on the number of places of critical vertices that can be covered in $U$ and $W$ respectively. 

The idea of the algorithm is the same again, but here, we use the fact that the stable matching problem can be solved in quadratic with Fleiner's algorithm (linear time in case of laminar matroids), even if each vertex has matroid constraints. 

Hence, the steps are:

1. Create an instance $I'$ with $G'=(U,W;E')$ of stable matching problem with matroid constraints and strict preferences by making parallel copies of each edge and create strict preferences over the created edges. 

2. Run Fleiner's algorithm to obtain a stable \mmatching\ $M'$ in the new instance $I'$.

3. Take the projection $M$ of $M'$ to $I$ by taking an edge $e$ inside $M$, whenever one of the parallel copies of $e$ was inside $M'$.

Let us start by describing the extended instance $I'$. For each edge $e=(u,w)\in E$ we create parallel edges as follows:
\begin{itemize}
    \item [--] We create copies $a(e)$, $b_0(e),b_1(e)$ and $c(e)$,
    \item [--] If $w\in C\cap W$ is critical and $e\in E_c$ is also critical, we create copies $x_1(e),\dots,x_{t+7}(e)$
    \item [--] If $u\in C\cap U$ is critical and $e\in E_c$ is also critical, we create copies $z_1(e),\dots, z_{s+7}(e)$
    \item [--] Finally, if $e\in E_c$ is critical and both $u,w\in C$ are critical, we create a copies $y_0(e),y_1(e)$
\end{itemize}

Then, we create the strict preferences according to the same rule as in the one-to-one matching case.
The matroids $\mathcal{M}'_v$ for the agents in the extended instance are created such that $M'\subset E'(v)$ is independent, if and only if it contains at most one parallel copy of each edge and its projection $M\subset E(v)$ is independent in $\mathcal{M}_v$. Therefore,
$M'\subseteq E'$ is a \mmatching\ if and only if it contains at most one copy of each edge and its projection $M\subseteq E$ is a \mmatching.

We proceed by proving the algorithm is a $\frac{3}{2}$-approximation even in this most general model.

\begin{theorem}
\cgkernel\ can be $\frac{3}{2}$-approximated in $\mathcal{O}((s+t)^2|E|^2)$ time even with arbitrary matroid constraints. If the matroid contraints are given by laminar matroids, then the running time is $\mathcal{O}((s+t)|E|)$.
\end{theorem}
\begin{proof}
We prove the theorem in three claims again. Let $M$ denote the output of the algorithm and $M'$ be its preimage in the extended instance $I'$.

\begin{claim}
$M$ is critical.
\end{claim}
\begin{claimproof}
Suppose for contradiction that there is a \mmatching\ $N$, such that $\sum_{v\in C}|M_c(v)|<\sum_{v\in C}|N_c(v)|$. Then, it must hold that either $N_c$ covers more places of critical vertices from $U$ or more from $W$ than $M_c$. Suppose that the first case holds.

We create a graph $H=(U',W',E_{M_c\cup N_c})$ by making $q(v)=|E(v)|$ copies $v^1,\dots,v^{q(v)}$ of each vertex $v\in U\cup W$. Then, for each edge $e=(u,w)\in M_c$ we add an edge $e'=(u^i,w^j)$ (for some copy $u^i$ of $u$ and $w^j$ of $w$) to $E_{M_c\cup N_c}$ such that these edges form a matching $M_c^H$ in $H$. As $|M\cap E(v)|\le q(v)$ $\forall v\in U\cup W$, this is clearly possible. Then, we add an edge $e'=(u^i,w^j)$ for each $e=(u,w)\in N_c$ such that 
\begin{enumerate}
    \item These edges form a matching $N_c^H$ in $H$,
    \item If $M_c(u)+e\in \cI_u$, then $u^i$ is not covered by $M_c^H$,
    \item If $M_c(w)+e\in \cI_w$, then $w^j$ is not covered by $M_c^H$,
    \item If $M_c(u)+e\notin \cI_u$, then $u^i$ is covered by a copy $f'\in M_c^H$ of an edge $f\in M_c$, such that $M_c(u)-f+e\in \cI_u$,
    \item If $M_c(w)+e\notin \cI_w$, then $w^j$ is covered by a copy $f'\in M_c^H$ of an edge $f\in M_c$, such that $M_c(w)-f+e\in \cI_w$.
\end{enumerate}

We can choose the endpoint in $U'$ for each edge such that 2 and 3 hold, because we have $q(v)=|E(v)|$ many options which is always sufficient. We can also choose the enpoint in $W'$ such that 4 and 5 hold, because if $M_c(w)+e\notin \cI_w$ or $M_c(u)+e\notin \cI_u$, then the fundamental cycle of $e$ in $M$ is included in $M_c$ and by Lemma \ref{matroidfact}, the Hall-property is satisfied, so we can match the edges to the vertex copies such that 4 and 5 will hold.

By our assumption, there must be a component $P$ in $N_c^H\cup M_c^H$ such that $N_c^H$ matches strictly more copies of critical vertices than $M_c^H$ in $U'$. We can conclude that there is a vertex $u_{k_1}\in U\cap C$ such that a copy $u_{k_1}^{i_1}$ is covered by $N_c^H$ with an edge $(u_{k_1}^{i_1},w_{l_1}^{j_1})$, but not by $M_c^H$. Let $e_1=(u_{k_1},w_{l_1})$. Hence, $M_c(u_{k_1})+e_1\in \cI_{u_{k_1}}$.

As $u_{k_1}$ is critical, we know that the copies $z_1(e_1),\dots, z_{s+7}(e_1)$ exist. As $M_c(u_{k_1})+e_1\in \cI_{u_{k_1}}$, we know that either $M(u_{k_1})+e\in \mathcal{M}_{u_{k_1}}$ or otherwise there is an edge $f_1\in M\setminus M_c$ such that $M(u_{k_1})-f_1+e_1 \in \cI_{u_{k_1}}$. In particular, we know that if there is such an edge $f_1$, then only $a,b_0,b_1,c$ copies exist of $f$. Hence, by the fact that $z_4(e_1)$ does not block $M'$ we get that $M(w_{l_1})+e_1\notin \cI_{w_{l_1}}$ and that for each edge $f$ in the fundamental circuit of $e_1$ in $\mathcal{M}_{w_{l_1}}$, $z_j(f)\in M'$ with $j\le 4$. This implies that for each edge $f=(u,w_{l_1})\in M(w_{l_1})$ in the fundamental cycle of $e_1$, $f\in E_c$ and $u\in C$. By property 5, in $M_c^H$ $w_{l_1}^{j_1}$ is matched to some $u_{k_2}^{i_2}$ ($k_2$ can be the same as $k_1$ but it is not important for the analysis) with $u_{k_2}\in C$ and $f_2=(u_{k_2},w_{l_1})\in E_c$. 
By our assumption on the component $P$ ($N_c^H$ matches more copies of critical vertices), $u_{k_2}^{i_2}$ is matched in $N_c^H$ to a vertex $w_{l_2}^{j_2}$. Let $e_2=(u_{k_2},w_{l_2})$.

By property 4, we know that $M_c(u_{k_2})+e_2\notin \cI_{u_{k_2}}$, so the fundamental cycle of $e_2$ in $\mathcal{M}_{u_{k_2}}$ is included in $M_c(u_{k_2})$ and $M_c(u_{k_2})+e_2-f_2\in \cI_{u_{k_2}}$.
As any $z_5$ copy is better for $u_{k_2}$ in $I'$ than any $z_j$ copy with $j\le 4$, but $z_5(e_2)$ exists and does not block, we have $M(w_{l_2})+e_2\notin \cI_{w_{l_2}}$ and for each $f\in M(w_{l_2})$ in the fundamental cycle of $e_2$ in $\mathcal{M}_{w_{l_2}}$, $z_j(f)\in M'$ with some $j\le 5$. Therefore, for each such edge $f=(u,w_{l_2})$, $f\in E_c$ and $u\in C$. By propety 5, the copy of $w_{l_2}$ covered by the image of $e_2$ in $N_c^H$ is also covered by a copy of one of these edges $f$, say $f_3=(u_{k_3},w_{l_2})$. By our assumption on $P$, the copy of $u_{k_3}$ matched by $f_3$ is also covered by $N_c^H$ with an edge $(u_{k_3}^{i_3},w_{l_3}^{j_3})$. Let $e_3=(u_{k_3},w_{l_3})$. Using property 4, we get that $f_3$ is in the fundamental cycle of $e_3$ in $M(u_{k_3})$.

By iterating this argument, we get that there must be more than $s$ vertices of $U'$ that are copies of critical vertices in $U$ and are covered by $N_c^H$, which is a contradiction, as at most $s$ places of the critical vertices in $U$ can be covered by a \mmatching.

The second case is analogous. 
\end{claimproof}

\begin{claim}
The output \mmatching\ $M$ by the algorithm is \cgstab.

\end{claim}
\begin{claimproof}
Suppose for contradiction that there is a \cgblocking\ edge $e=(u,w)$ to $M$ with $f\in M(u)\cup \{ \emptyset \}$ and $g\in M(w)\cup \{ \emptyset \}$.

If $M(u)+e\in \cI_u$, then $M(w)+e\notin \cI_w$, so $g=(u',w)\in M(w)\cup \{ \emptyset \}$ is not $\emptyset$. As $e$ \cgblocks\ and $M$ is critical we get that $e\in E_c$ if and only if $g\in E_c$ and if yes, then $u'\in C$ if and only if $u\in C$. In particular, we get that the best copy of $e$ and $g$ in $I'$ for $w$ is the same. Also, this copy can be only $c,x_{t+7}$ or $z_1$ and all these copies are ranked according to $p_w()$. Therefore, this copy of $e$ blocks $M'$, as $e$ \cgblocks\ $M$, contradiction. The case when $M(w)+e\in \mathcal{M}_w$ is similar.

Hence, $M(u)=f=(u,w')$ and $M(w)=g=(u',w)$ is not $\emptyset$. Because $e$ \cgblocks, $p_u(e)\ge p_u(f)+\gamma_e^u$ and $p_w(e)\ge p_w(g)+\delta_e^w$ or $p_u(e)\ge p_u(f)+\delta_e^u$ and $p_w(e)\ge p_w(g)+\gamma_e^w$. 

If $f,g\notin E_c$, then $e\notin E_c$, as $M$ is critical. Hence, only $a,b_0,b_1,c$ copies exists of $e,f,g$. As $e$ \cgblocks\ $M$, we get that $b_0(e)$ or $b_1(e)$ blocks $M'$, contradiction.

Suppose that $f,g\in E_c$. Then, $e\in E_c$, $u,w\in C$ and $u',w'\notin C$, because $e$ \cgblocks. Hence, $x_i(f),z_j(g)$ do not exists for any $i,j$, but $y_0(e),y_1(e)$ exists. As $e$ \cgblocks\, one of $y_0(e),y_1(e)$ blocks $M'$, contradiction.

Suppose that $f\in E_c,g\notin E_c$. In particular, only $a(g),b_0(g),b_1(g),c(g)$ exists of $g$. Then, $e\in E_c$ as $e$ \cgblocks\ and $w'\in C$ if and only if $w\in C$. If $w,w'\in C$, then $x_2(e),x_3(e)$ exists and one of them blocks $M'$, contradiction. If $w,w'\notin C$, then $x_i(e),x_i(f)$ do not exists for any $i$ and neither does $y_0(e),y_0(f),y_1(e),y_1(f)$. Hence, $z_{s+5}(e)$ or $z_{s+6}(e)$ blocks $M'$, contradiction.

The case $f\notin E_c,g\in E_c$ is analogous. 
\end{claimproof}

\begin{claim}
    For any \cminstab\ \mmatching\ $N$ it holds that $|N|\le \frac{3}{2}|M|$
\end{claim}
\begin{claimproof}
    Suppose for contradiction that there is a \cgstab\ \mmatching\ $N$ such that $|N|>\frac{3}{2}|M|$. 

Let $N_i$ be a subset of $N\setminus M$ such that $M \cup N_i \in \cI_i$ and $|M \cup N_i|=|N|$ for each $i\in \{U,W\}$. The sets $N_U$ and $N_W$ are disjoint because $M'$ is an inclusionwise maximal common independent set in $I'$. 

Let $\mathcal{M}^*_i$ be the matroid obtained from $\mathcal{M}_i$ by deleting $E\setminus (M \cup N)$, contracting $(M\cap N) \cup N_i$, and truncating to the size of $M\setminus N$. That is, the ground set of $\mathcal{M}_i^*$ is $E_i^*=(M\setminus N)\cup (N\setminus (M\cup N_i))$ and $\mathcal{M}_i^*=\{\, I\subseteq E^*_i: I\subseteq M\cup N,~ I\cup (M\cap N)\cup N_i \in \mathcal{M}_i,~ |I|\leq |M\setminus N|\,\}$. In $\mathcal{M}_i^*$, the sets $M_i^*\coloneqq M\setminus N$ and $N_i^*\coloneqq N \setminus (M \cup N_i)$ are bases that are complements of each other.

We define a strict preference order $\succ^*_i$ on $E_i^*$ in the following way. The elements of $N \setminus (M \cup N_U \cup N_W)$ are worst (in arbitrary order). On the remaining elements, we define the preferences based on the strict preferences created in $I'$ over $E'$. We do this as follows. For an edge $e\in N_U\setminus M$, we take the best existing copy for $\mathcal{M}_W$. For an edge $e\in N_W\setminus M$, we take the best existing copy for $\mathcal{M}_U$. Finally, for $e\in M\setminus N$, we take the copy included in $M'$. Then, we order these edges (first just within the sets $E(v)\cap (M\cup N_U\cup N_W)$) according to the preference orders induced on these copies in $I'$ and then extend this order by appending the constructed individual orders within the sets $E(v)\cap (M\cup N_U\cup N_W)$ after each other arbitrarily.

In the ordered matroid $\mathcal{M}_i^*=(E^*,\mathcal{I}^*_i,\succ^*_i$), $M_i^*$ is an optimal base. Indeed, $e$ is the worst element in the fundamental cycle of $e$ in $M_i^*$ for every $e\in N_i^*\cap (N_U\cup N_W)$, as otherwise a copy of $e$ would block $M'$, contradiction. It is also clear for the elements in $N_i^*\setminus (N_U\cup N_W)$ by the definition of $\succ^*_i$. By Theorem \ref{thm:matching}, there is a perfect matching $P_i$ for $i=U,W$ between $M_i^*$ and $N_i^*$ such that $e \succ_i^* f$ and $f$ is in the fundamental cycle of $e$ in $N_i^*$ for every $(e,f) \in P_i$, where $e \in M_i^*$ and $f\in N_i^*$.

Since $|N|>1.5|M|$ implies $|N_U^*|=|N_W^*|>|M\setminus N|/2$, there is an element $e \in M\setminus N$ that is covered by both $P_U$ and $P_W$. Let $(e,f) \in P_U$, $(e,g) \in P_W$ with $f=(u_1,w_1),e=(u_1,w_2),g=(u_2,w_2)$.

Note that this means that if $e'$ is the copy of $e$ in $M'$, then $e'$ is better than the best existing copy of $f$ for $u_1$ and $e'$ is better than the best existing copy of $g$ for $w_2$. Also, $N+e-f\in \cI_U$ and $N+e-g\in \cI_W$.

Suppose that $f,g\notin E_c$. Then, $e\notin E_c$ and only $a,b_0,b_1,c$ copies exists of all three edges, because otherwise $N+e-f-g$ is a \mmatching\ covering more places of critical vertices with critical edges, than $N$, contradicting that $N$ is critical. As $a(f)$ is not better than $e'$, $b_0(e),b_1(e)$ or $c(e)\in M'$. As $c(g)$ is not better than $e'$, $b_0(e),b_1(e)$ or $a(e)\in M'.$ Hence, $b_0(e)$ or $b_1(e)\in M'$. If $b_0(e)\in M'$, then $p_u(e)\ge p_u(f)+\gamma_e^u$ and $p_w(e)\ge p_w(g)+\delta_e^w$, so $e$ \cgblocks\ $N$ (combining with the fact that $e$ and $f\cup g$ both match 0 critical nodes), contradiction. The other case implies $p_u(e)\ge p_u(f)+\delta_e^u$ and $p_w(e)\ge p_w(g)+\gamma_e^w$, so $e$ \cgblocks\ $N$, contradiction again.

Next suppose that $f,g\in E_c$. If $x_1(f)$ exists, then as $e'$ is better than $x_1(f)$ for $u_1$, we get that $e'\in \{ x_1(e),x_2(e),x_3(e)\}$. Hence, 
as either $x_{t+7}(g)$ or $z_1(g)$ exists by $g\in E_c$, it cannot hold that $e'$ is better than the best existing copy of $g$, contradiction. Similarly we get that $z_1(g)$ does not exist, so $u_1,w_2\in C$, $u_2,w_1\notin C$. As $e'$ is better than $z_{s+7}(f)$ for $u_1$ and better than $x_{t+7}(g)$ for $w_2$ we get that $e'\in \{ y_0(e),y_1(e)\}$. In particular, $e\in E_c$ and $p_u(e)\ge p_u(f)+\gamma_e^u$,$p_w(e)\ge p_w(g)+\delta_e^w$ or $p_u(e)\ge p_u(f)+\delta_e^u$, $p_w(e)\ge p_w(g)+\gamma_e^w$. Hence, $e$ is \cgblocking\ $N$, contradiction.

Suppose that $f\in E_c,g\notin E_c$. If $w_1\in C$, then as $e'$ is better for $u_1$ than $x_1(f)$, $e'\in \{ x_1(e),x_2(e),x_3(e)\}$. It is also better than $c(g)$ for $w_2$, so $e'$ is either $x_2(e)$ or $x_3(e)$. In particular, $e\in E_c$ and $w_2\in C$, which means that $e$ \cgblocks\ $N$, contradiction. If $w_1\notin C$, then by the facts that $e'$ is better than $z_{s+7}(f)$ for $u_1$ and better than $c(g)$ for $w_2$, we get that $e'$ cannot be an $a,b_0,b_1$ or $c$ type copy.  Hence, $e\in E_c$. This implies that $w_2\notin C$, as otherwise $N+e-f-g$ matches more places of critical vertices with critical edges, than $N$, contradiction. So there is no $y_0(e),y_1(e)$ copy. Combining this, we obtain that either $z_{s+5}(e)\in M'$ or $z_{s+6}(e)\in M'$. In both cases, it follows that $e$ \cgblocks\ $N$, contradiction. 

The case $f\notin E_c,g\in E_c$ is similar.
\end{claimproof}

The running time follows from the fact that Fleiner's algorithm runs in time $\mathcal{O}(|E'|^2)=\mathcal{O}((s+t)^2|E|^2)$-time. In the case, when the matroids $\mathcal{M}_v$ are simple laminar matroids, coming from for example vertex capacities $q(v)$, the matroids in $I'$ are laminar, hence the running time of Fleiner's algorithm can be reduced to $\mathcal{O}((s+t)|E|)$ and so is our algorithm's.
The theorem now follows.
\end{proof}


\begin{remark}
    Even this algorithm is straghtforward to extend to work for the case when $k$ $(\gamma_e^v)_i$ numbers are given, but the number of required copies make the analysis more technical and tedious. 
\end{remark}

\section{Conclusion and Future Work}

In this paper we demonstrated the robust usefulness and generality of the edge duplicating technique and provided simple algorithms for many \smti\ generalizations at the same time. As there are probably even more cases, where this technique could provide a nice solution, there is a huge potential of applying it to other lesser known or yet to be introduced models.

\section{Acknowledgements}
The author was supported by the Hungarian Scientific Research Fund, OTKA, Grant No. K143858, by the Momentum Grant of the Hungarian Academy of Sciences, grant number 2021-1/2021 and by the Ministry of Culture and Innovation of Hungary from the National Research, Development and Innovation fund, financed under the KDP-2023 funding scheme (grant number C2258525).

\printbibliography
\end{document}